\documentclass[11pt]{article}
\usepackage[twoside,a4paper,margin=1in]{geometry}
\usepackage{microtype}
\usepackage{graphicx,url,amsmath,amsthm,amsfonts,amssymb,subfigure,bbm,bm}

\usepackage[dvips]{epsfig}
\usepackage{thm-restate}
\usepackage{cutwin}
\usepackage{float}
\usepackage{color}

\theoremstyle{plain}
\newtheorem{theorem}{Theorem}
\newtheorem{proposition}[theorem]{Proposition}

\newtheorem{corollary}[theorem]{Corollary}
\newtheorem{lemma}[theorem]{Lemma}

\newcommand{\R}{\mathbb{R}}
\newcommand{\Z}{\mathbb{Z}}

\newcommand{\V}{\mathcal V}

\newcommand{\RP}{\mathbb{R}P}

\newcommand{\uni}{universal shortest path metric}

\newcommand{\changed}[1]{#1}

\newcommand{\blue}[1]{{\color{blue}\bf #1}}

\newcommand{\red}[1]{{\color{red}\bf #1}}
\newcommand{\magenta}[1]{{\color{magenta}\bf #1}}
\newcommand{\green}[1]{{\color{green}\bf #1}}
\newcounter{sideremark}
\newcommand{\marrow}{\stepcounter{sideremark}\marginpar{$\boldsymbol{\longleftarrow\scriptstyle\arabic{sideremark}}$}}

\newif\ifappendix
\appendixtrue 

\newif\ifcomments
\commentstrue 
\ifcomments
\newcommand{\martin}[1]{{\vskip 5pt\textsf{\magenta{*** (Martin) \marrow #1}\vskip 5pt}}
}
\newcommand{\alfredo}[1]{{\vskip 5pt\textsf{\blue{*** (Alfredo) \marrow #1}\vskip5pt}}
}

\newcommand{\arnaud}[1]{{\vskip 5pt\textsf{\red{*** (Arnaud) \marrow #1}\vskip 5pt}}
}

\newcommand{\vojta}[1]{{\vskip 5pt\textsf{\green{*** (Vojta) \marrow #1}\vskip 5pt}}
}
\else
\newcommand{\martin}[1]{}
\newcommand{\alfredo}[1]{}
\newcommand{\arnaud}[1]{}
\newcommand{\vojta}[1]{}
\fi

\DeclareMathOperator{\dist}{dist}

\DeclareMathOperator{\OUT}{OUT}

\title{Shortest path embeddings of graphs on surfaces\thanks{The project was
    partially supported by the Czech-French collaboration project EMBEDS (CZ: 7AMB15FR003, FR: 33936TF). 
    The research of A.~H. was funded by GUDHI, geometric understanding in higher dimensions.
    The
    research of A.~dM. leading to these results has
    received funding from the People Programme (Marie Curie Actions)
    of the European Union's Seventh Framework Programme
    (FP7/2007-2013) under REA grant agreement n$^{\circ}$ [291734]. V.~K.~and M.~T. were
partially supported by the project CE-ITI (GA\v{C}R P202/12/G061).}}

\author{Alfredo Hubard\thanks{INRIA Sophia-Antipolis and Universit\'e
    Paris-Est Marne-la-Vall\'ee, France,
  \protect\url{alfredo.hubard@inria.fr}}\and Vojt\v{e}ch
  Kalu\v{z}a\thanks{Department of Applied Mathematics, Charles University,
  Malostransk\'{e} n\'{a}m\v{e}st\'{\i} 25, Prague 1, Czech Republic,
\protect\url{kaluza@kam.mff.cuni.cz}}\and Arnaud de Mesmay\thanks{CNRS,
Gipsa-Lab, 11 rue des Math\'ematiques, 38 400 Saint Martin d'H\`{e}res,
France, \protect\url{arnaud.de-mesmay@gipsa-lab.fr}}\and Martin
Tancer\thanks{Department of Applied Mathematics and Institute of Theoretical
Computer Science, Charles University, Malostransk\'{e} n\'{a}m\v{e}st\'{\i} 25,
Prague 1, Czech Republic, \protect\url{tancer@kam.mff.cuni.cz}} }

\begin{document}

\maketitle

\begin{abstract}
  The classical theorem of F\'{a}ry states that every planar graph can be
  represented by an embedding in which every edge is represented by a straight
  line segment. We consider generalizations of F\'{a}ry's theorem to 
  surfaces equipped with Riemannian metrics. In this setting, we require that
  every edge is drawn as a shortest path between its two endpoints and we
  call an embedding with this property a \emph{shortest path embedding}. The
  main question addressed in this paper is whether given a closed surface
  $S$, there exists a Riemannian metric for which every topologically
  embeddable graph admits a shortest path embedding. This question is also
  motivated by various problems regarding crossing numbers on surfaces.

  We observe that the round metrics on the sphere and the projective plane have
  this property. We provide flat metrics on the torus and the Klein bottle
  which also have this property.

  Then we show that for the unit square flat metric on the Klein bottle there
  exists a graph without shortest path embeddings. We show, moreover, that for
  large $g$, there exist graphs $G$ embeddable into the orientable surface of
  genus $g$, such that with large probability a random hyperbolic metric does
  not admit a shortest path embedding of $G$, where the probability measure is
  proportional to the Weil--Petersson volume on moduli space.

  Finally, we construct a hyperbolic metric on every orientable surface $S$ of
  genus $g$, such that every graph embeddable into $S$ can be embedded so that 
  every edge is a concatenation of at most $O(g)$ shortest paths.




 
\end{abstract}



\bigskip



\section{Introduction}
\subparagraph*{F\'{a}ry's theorem and joint crossing numbers.} A famous theorem of F\'{a}ry~\cite{f-oslrp-48} 
states that any simple planar graph can be
embedded so that edges are represented by straight line segments. In this article we investigate analogues of this theorem in the context of graphs embedded into surfaces. We focus on the following problem:
Given a surface $S$, is there a metric on $S$ such that every graph embeddable
into $S$ can be embedded so that edges are represented by
shortest paths?

We call such an embedding a \textit{shortest path
embedding}, and such a metric a \emph{\uni}.\footnote{We do not require that these shortest
paths are unique but as we will see later on, in the case of our positive
results, i.e., Theorem~\ref{t:shortest_paths} and~\ref{T:concat},
the uniqueness of the shortest paths can be obtained as well.).}

Before being enticed by this question, we were motivated to consider it by a number of problems
involving joint embeddings of curves or graphs on surfaces
arising from seemingly disparate settings. The literature on the subject goes back at least 15 years
with Negami's work related to diagonal flips in
triangulations~\cite{n-cngepc-01}. He conjectured that there exists a
universal constant $c$ such that for any pair of graphs $G_1$ and
$G_2$ embedded in a surface $S$, there exists a homeomorphism
$h:S\rightarrow S$ such that $h(G_1)$ and $G_2$ intersect transversely
at their edges and the number of edge crossings satisfies
$cr(h(G_1),G_2) \leq c |E(G_1)|\cdot |E(G_2)|$.

Recently, on one hand,
Matou\v{s}ek, Sedgwick, Tancer, and
U. Wagner~\cite{mstw-utsnc-13,mstw-e3sd-14}, working on decidability
of embeddability of $2$-complexes into $\mathbb{R}^3$ and on the other
hand, Geelen, Huynh, and Richter~\cite{ghr-epgm-13}, in a quest for
explicit bounds for graph minors, were faced with a similar question
and provided bounds for related problems. 
Joint crossing number type problems are dually
equivalent to problems of finding a graph with a specific pattern
within an embedded graph while bounding the multiplicity of the edges
used. This is a fundamental concern of computational topology of
surfaces where one is interested in finding objects with a fixed
topology and minimal combinatorial complexity, e.g., short canonical
systems of loops~\cite{lpvv-ccpso-01}, short pants
decompositions~\cite{cl-opdsh-07} or short octagonal
decompositions~\cite{ce-tnpcs-10}; see
also~\cite{chm-dsidts-15}.


Negami provided the upper bound $cr(h(G_1),G_2) \leq c g |E(G_1)|\cdot
|E(G_2)|$, and despite subsequent
discoveries~\cite{ab-tms-01,rs-tmlrs-05}, his conjecture is still
open.  
In a paper that refines Negami's work~\cite{rs-tmlrs-05}, Richter and
Salazar wrote \textit{``this [conjecture] seems eminently reasonable:
  why should two edges be forced to cross more than once?''}. The
connection with our work is that if two graphs are embedded
transversally by shortest path embeddings, then indeed no two edges
cross more than once, since otherwise one of them could be
shortcut. \changed{In particular, a proof that every surface admits a universal
shortest path metric would
  imply Negami's conjecture, actually even if we allowed to subdivide each edge
of the embedded graph constantly many times.}

We note that prior to our work, Schaefer~\cite[paragraph on
  Geodesic crossing numbers]{s-gcnvs-13} had considered similar
questions, mainly for drawing edges of a graph by geodesics. We
provide the details below including answers to some of Schaefer's
questions. We also note that our methods easily yield a new proof of
Negami's theorem for orientable surfaces; see
Corollary~\ref{c:negami}.






Beyond crossing numbers, the existence or non-existence of shortest path universal
metrics might be relevant in curvature free and extremal Riemannian geometry.

\subparagraph*{Related work.}Various results in graph drawing~\cite{t-hgdv-13} revolve around generalizing F\'ary's theorem to find drawings of graphs with additional constraints, for instance drawing the edges with polylines with few bends. On the other hand, only few
extensions to graphs embedded in surfaces are known. Two classical
avatars of F\'ary's theorem in the plane are of relevance to our work:
Tutte's barycentric embedding theorem~\cite{t-hdg-63} and the
Koebe-Andreev-Thurston circle packing theorem (see, for example, the
book of Stephenson~\cite{s-icptda-05}). Both have been generalized to
surfaces, providing positive answers to the following questions:


\begin{enumerate}
\item Given a surface $S$, a metric $m$, and a graph $G$ embeddable into
  $S$, can we embed the graph $G$ so that every edge is represented by
  a geodesic with respect to $m$?
\item Given a graph $G$ embeddable into $S$, does there exist a metric
  $m$ on $S$ so that $G$ embeds into $S$ with shortest paths?
\end{enumerate}



The first question was considered by \changed{Schaefer~\cite{s-gcnvs-13};} a positive answer for many metrics had been previously given by Y. Colin de
Verdi\`ere~\cite{c-crgts-91} who generalized Tutte's barycentric embedding
approach using a variational principle. The idea behind this approach is to start with a
topological embedding of the graph, replace the edges by springs, and
let the system reach an equilibrium. Y. Colin de Verdi\`ere proved that for
any metric of non-positive curvature, the edges become geodesics with
disjoint interiors when the system reaches stability; moreover, this
embedding is essentially unique within its homotopy class. However,
geodesics need not be shortest paths, and two geodesics can intersect
an arbitrarily large number of times, see Figure~\ref{F:examples}. Yet, these examples do not provide a negative
answer to the second question, or to our main question, since we could
change the embedding by a homeomorphism of the torus (thus even
preserving the combinatorial map) to obtain a shortest path embedding.

\begin{figure}
  \centering
  \def\svgwidth{.9\textwidth}
	
  \begingroup%
  \makeatletter%
  \providecommand\color[2][]{%
    \errmessage{(Inkscape) Color is used for the text in Inkscape, but the package 'color.sty' is not loaded}%
    \renewcommand\color[2][]{}%
  }%
  \providecommand\transparent[1]{%
    \errmessage{(Inkscape) Transparency is used (non-zero) for the text in Inkscape, but the package 'transparent.sty' is not loaded}%
    \renewcommand\transparent[1]{}%
  }%
  \providecommand\rotatebox[2]{#2}%
  \ifx\svgwidth\undefined%
    \setlength{\unitlength}{3540.03641131bp}%
    \ifx\svgscale\undefined%
      \relax%
    \else%
      \setlength{\unitlength}{\unitlength * \real{\svgscale}}%
    \fi%
  \else%
    \setlength{\unitlength}{\svgwidth}%
  \fi%
  \global\let\svgwidth\undefined%
  \global\let\svgscale\undefined%
  \makeatother%
  \begin{picture}(1,0.19588061)%
    \put(0,0){\includegraphics[width=\unitlength,page=1]{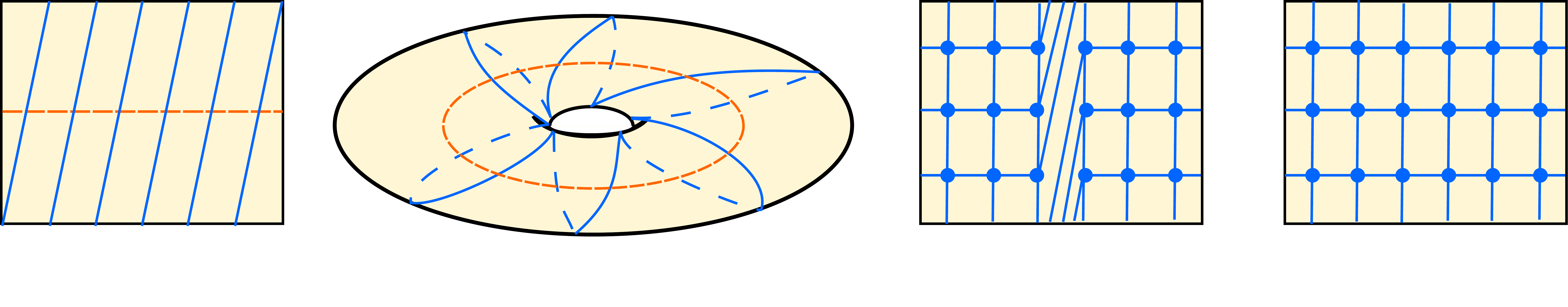}}%
    \put(0.0698879,0.00103062){\color[rgb]{0,0,0}\makebox(0,0)[lb]{\smash{a.}}}%
    \put(0.36559148,0.000128){\color[rgb]{0,0,0}\makebox(0,0)[lb]{\smash{b.}}}%
    \put(0.67357863,0.00103062){\color[rgb]{0,0,0}\makebox(0,0)[lb]{\smash{c.}}}%
    \put(0.90214773,0.000128){\color[rgb]{0,0,0}\makebox(0,0)[lb]{\smash{d.}}}%
  \end{picture}%
\endgroup%

  \caption{a. and b. Two geodesics crossing many times. c. A grid embedded in a torus with geodesics. d. A reembedding of this grid with shortest paths.}
  \label{F:examples}
\end{figure}
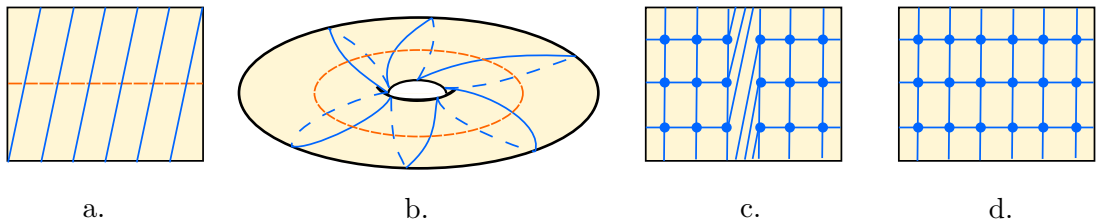





The second question also has a positive answer, which can be proved via a generalization of the circle packing theorem to closed surfaces~\cite{s-icptda-05}. Namely, for every triangulation $T$ of a surface, there exists a
metric of constant curvature so that $T$ can be represented as the
contact graph of a family of circles. The representation of the triangulation that places a vertex at the center of its corresponding circle is an embedding with shortest paths. Such a representation can be computed efficiently and can be used as a tool for representing graphs on surfaces~\cite{m-dghp-99}. However, the metric is determined by the triangulation, which makes this approach ill-suited for our purpose.

\subparagraph*{Our results.} Our objective here is a mix of these last two results. 
On the one hand, we require shortest paths and not geodesics, on the
other hand, we want a single metric for each surface and not one which
depends on the triangulation. We will also consider the relaxation of
our problem where we are allowed to use concatenations of shortest
paths: we say that a metric is a \textit{k-universal shortest path metric} if every topologically
embeddable graph can be represented by an embedding in which edges are
drawn as concatenations of $k$ shortest paths. This is akin to various
problems in graph drawing where graphs are embedded with polylines
with a bounded number of bends instead of straight
lines~\cite{ek-sepgfb-05,t-eggmnb-87}.

Our results focus on Riemannian metrics of constant curvature, and our techniques are organized by the sign of the curvature. We first observe that for the sphere and the projective plane, since there is a unique Riemannian metric of curvature $1$, the circle
packing approach applies to all graphs. Then, with the aid of \emph{irreducible triangulations}, we provide flat metrics (i.e., of zero curvature) on the torus and the Klein bottle for which every graph
admits a shortest path embedding.


\begin{restatable}{theorem}{tshortestpaths}\label{t:shortest_paths}
The sphere $S^2$, the projective plane $\RP^2$, the torus $T^2$, and the Klein bottle
$K$ can be endowed with a \uni.
\end{restatable}

This result could lead to the idea that shortest path embeddings can
be achieved for any metric, i.e., that every metric is a \uni. We prove
that this is not the case already for the unit square flat metric
on the Klein bottle (arguably the first example to consider).

\begin{restatable}{theorem}{tklein}\label{t:klein}
Let $K$ denote the Klein bottle endowed with the unit square flat
metric on the polygonal scheme $aba^{-1}b$. Then there exists a graph embeddable into $K$
which cannot be embedded into $K$ so that the edges are shortest paths.
\end{restatable}

In higher genus, the number of irreducible triangulations is too large to check
all cases by hand. Hyperbolic surfaces of large genus are hard to comprehend,
but the probabilistic \changed{point of view} allows us to show that if there
\changed{exist} universal shortest path metrics of constant curvature $-1$ at all, their fraction tends to $0$ as the genus tends to infinity.

\begin{restatable}{theorem}{Tnegative}\label{negative}
For any $\varepsilon>0$, with probability tending to $1$ as $g$ goes to infinity, a random
hyperbolic metric is not a $O(g^{1/3-\varepsilon})$-\uni. In particular, with probability tending to $1$ as $g$ goes to infinity, a random hyperbolic metric is not a universal shortest path metric.
\end{restatable}

Here the probability measure on the space hyperbolic surfaces is proportional to the Weil--Petersson volume, see Section~\ref{S:pants}. Our proof is an application of deep
results on this volume by Mirzakhani~\cite{m-gwpvrh-13} and Guth, Parlier, and Young~\cite{gpy-pdrs-11}.

For a given graph $G$ and a metric $m$ on $S$, Schaefer~\cite{s-gcnvs-13}
defines the \emph{geodesic crossing number} of $G$ as the minimal number of
crossings of any drawing of $G$ in $S$ in which edges are represented by
geodesics. Schaefer asks if this definition is equivalent to the analogous
definition with shortest paths instead of geodesics. Notice that the examples
in Theorems 2 and 3 have nonpositive curvature, hence, combined with the
aforementioned result of Y. Colin de Verdi\`ere imply that some graphs have geodesic crossing number zero but \emph{shortest path crossing number} nonzero, answering Schaefer's question. 

For genus $g>1$ we do not know if there exist shortest path universal
metrics. But relaxing the question to concatenations of shortest paths
and combining ideas from hyperbolic geometry and computational
topology, we provide for every orientable surface of genus $g$ an
$O(g)$-universal shortest paths metric. The proof relies on the
octagonal decompositions of \'E. Colin de Verdi\`ere and
Erickson~\cite{ce-tnpcs-10} and a variant of the aforementioned
theorem of Y. Colin de Verdi\`ere~\cite{c-crgts-91}.


\begin{restatable}{theorem}{Tconcat}\label{T:concat}
For every $g>1$, there exists an $O(g)$-universal shortest path
hyperbolic metric $m$ on the orientable surface $S$ of genus $g$.
\end{restatable}

In this article we focused on Riemannian metrics of constant
curvature, but we remark that both of our last results also hold in some setting of piecewise-Euclidean metrics as
well. For the upper bound, it suffices to replace hyperbolic hexagons
with Euclidean ones, and the rest of the proof works similarly. 
The lower bound can be derived following the heuristic strong parallels between the
Weil--Petersson volume form on moduli space and the counting measure on
the space of $N=4g$ Euclidean triangles randomly glued together. In
particular the results that we use have analogs in this latter
space: see Brooks and Makover~\cite{bm-rcrs-04} and the second half of
the article of Guth, Parlier, and Young~\cite{gpy-pdrs-11}.



We have stated our results for graphs in this introduction. We note
that one could consider the problem of shortest path embeddings for a
graph with a fixed embedding up to a homeomorphism of the surface
(i.e., for a combinatorial map), which is more in the spirit of
Negami's conjecture. Our positive results can be stated in this
stronger version; i.e., in our proofs the map is preserved. Our
negative results would be weaker if the map had to be preserved, and
in fact the proofs deal firstly with the statements for maps and then
we derive the analog for graphs with some extra work. 

\subparagraph*{Open questions.} The main open question is the existence of universal shortest path
metrics, or $O(1)$-universal shortest path metrics. Natural candidates
for these are given by certain celebrated extremal metrics like the ones occurring as lower bounds for Gromov's
systolic inequality \cite{bs-pmrsl-94,gro-83}.

Lists of irreducible triangulations exist for the double torus and the non-orientable surface of genus up to four~\cite{s-gits-06}. While the numbers are too big to be investigated by hand as we did for the torus and the Klein bottle, it may be possible to investigate some computerized approach to test their shortest path embeddability for some well chosen hyperbolic metric.

Our Theorem~\ref{T:concat} only deals with orientable surfaces. A similar approach might work for non-orientable surfaces as well, the key issue being to generalize the octagonal decompositions of \'E. Colin de Verdi\`ere and
Erickson~\cite{ce-tnpcs-10} to the non-orientable setting. We leave this as an open problem.

\subparagraph*{Outline.} After introducing the main definitions in Section~\ref{S:prelim}, we will prove Theorems~\ref{t:shortest_paths}, \ref{t:klein}, \ref{negative}, and~\ref{T:concat} in Sections~\ref{S:shortest-paths}, \ref{S:klein}, \ref{S:pants}, and~\ref{S:concat}, respectively.

\section{Preliminaries}\label{S:prelim}

In this article we only deal with compact surfaces without
boundaries. By the classification theorem, these are characterized by
their \textit{orientability} and their \textit{genus}, generally
denoted by $g$. Orientable surfaces of genus $0$ and $1$ are
respectively the \textit{sphere} $S^2$ and the \textit{torus} $T^2$,
while non-orientable surfaces of genus $1$ and $2$ are the
\textit{projective plane} $\mathbb{R}P^2$ and the \textit{Klein
  bottle} $K$. The orientable surface of genus $g$ is denoted by $S_g$. The \textit{Euler genus} is equal to the genus for non-orientable surfaces and equals twice the genus for orientable surfaces.

By a \textit{path} on a surface $S$ we mean a continuous map $p: [0,1]
\rightarrow S$, and a \textit{closed curve} denotes a continuous map
$\gamma: S^1 \rightarrow S$. These are \textit{simple} if they are
injective.  We will be using occasionally the notions of
\textit{homotopy}, \textit{homology}, and \textit{universal cover}, we
refer to Hatcher~\cite{h-at-02} for an introduction to these
concepts.
All the graphs that we consider in this paper are \textbf{simple}
graphs unless specified otherwise, i.e., loops and multiple edges are
disallowed. An \textit{embedding} of a graph $G$ into a surface $S$
is, informally, a crossing-free drawing of $G$ on $S$. We refer to
Mohar and Thomassen~\cite{mt-gs-01} for a thorough reference on graphs
on surfaces, and only recall the main definitions. A graph embedding
is \textit{cellular} if its faces are homeomorphic to open
disks. Euler's formula states that $v-e+f=2-g$ for any graph with $v$
vertices, $e$ edges, and $f$ faces cellularly embedded in a surface
$S$ of Euler genus $g$. When the graph is not cellularly embedded,
this becomes an inequality: $v-e+f \geq 2-g$. A \textit{triangulation}
of a surface is a cellular graph embedding such that all the faces are adjacent to three edges. 
\changed{An \textit{isomorphism} between two triangulations is a bijection between the vertices, edges and faces that respects incidences}. By a slight abuse of language, we will
sometimes refer to an embedding of a triangulation, by which we mean
an embedding of its underlying graph which is homeomorphic to the
given triangulation.
A \textit{pants decomposition} of an \changed{orientable} surface
$S$ is a family of disjoint curves $\Gamma$ such that cutting $S$
along all of the curves of $\Gamma$ gives a disjoint union of pairs of
pants, i.e., spheres with three boundaries. Every orientable surface except the
sphere and the torus admits a pants decomposition
with $3g-3$ closed curves and $2g-2$ pairs of pants.
Note that all the pants
decompositions are not topologically the same, i.e., are not related by a
self-homeomorphism of the surface. A class of pants decompositions
equivalent under such homeomorphisms will be called the (topological)
\textit{type} of the pants decomposition. We say that an embedding
$f:G \rightarrow S$ \textit{contains a pants decomposition} if there
exists a subgraph $H \subseteq G$ such that $f:H \rightarrow S$ is a
pants decomposition of $S$.

In this article, we will also be dealing with notions coming from
Riemannian geometry, we refer
to the book of do Carmo for more background~\cite{c-rg-92}. By a
\textit{metric} we always mean a Riemannian metric, which associates
to every point of a surface the \textit{curvature} at this point. The
Gauss--Bonnet theorem ties geometry and topology; it implies that the
sign of a metric of constant curvature that a topological surface
accepts is determined solely by its Euler genus.  

 A Riemannian metric induces a length functional on paths and
closed curves. A path or a closed curve is a
\textit{geodesic} if the functional is locally minimal. Shortest paths between
two points are global minima of the length functional. Unlike in the plane, geodesics are not,
in general, shortest paths; in addition, neither geodesics nor shortest paths are unique in general. 
If we have a shortest path embedding of a graph where every edge is drawn as
the
unique shortest path between its endpoints, we speak of \textit{shortest
paths embedding with uniqueness}.


\section{Shortest path embeddings for low genus surfaces}\label{S:shortest-paths}



%
%

\tshortestpaths*

In the theorem above, for $S^2$ and $\RP^2$ we use the round metric of
positive constant curvature scaled to $1$. In the case of torus we use
the flat metric obtained by the identification of the opposite edges
of the square. In the case of the Klein bottle we can show that an
analogous result fails with the flat square metric on the polygonal
scheme $aba^{-1}b$, as we will see in Section~\ref{S:klein}. But we
can get the result for the metric obtained by the identification of
the edges of a rectangle of dimensions $1 \times b$ where $b =
\sqrt{4/3} + \varepsilon$ for some small $\varepsilon > 0$. (The edges
of length $1$ are identified coherently, whereas the edges of length
$b$ are identified in opposite directions.)

In all cases we can get shortest path embeddings with uniqueness. Actually, for
the torus and the Klein bottle, uniqueness will be a convenient assumption for
inductive proofs.




\subparagraph*{The sphere and the projective plane.} 



By the circle packing theorem any planar graph can be represented as the
contact graph of a circle packing on the sphere (endowed with the standard
round metric)~\cite[Theorem~4.3]{s-icptda-05}. \changed{On the sphere each circle is the
boundary of a cap (a metric ball), and by the center of the circle we mean the
center of the corresponding cap.} It is easy to see that drawing each edge $(u,v)$
of a contact graph, by the shortest path between the centers the circles
corresponding to $u$ and $v$ is an embedding. Since these are shortest paths,
this proves Theorem~\ref{t:shortest_paths} for $S^2$. 

\bigskip

For the projective plane, a similar circle packing theorem follows
from the spherical case. Since we could not find a reference in
the literature we include a proof here. \\

\changed{Henceforth, the sphere and the projective plane are always endowed with their usual spherical metrics (of constant curvature). For any circle packing $P \subset S^2$, consider its contact graph, which we denote by $C(P)$, together with the embedding of $C(P)$ to $S^2$ in which the edge corresponding to touching circles $C_v$ and $C_w$ is drawn by the geodesic between the centers of the circles $C_v$ and $C_u$ in $S^2$. We will only consider packings $P$ for which this embedding of $C(P)$ is a triangulation of $S^2$, and we call the corresponding triangulation the \emph{carrier} of $P$. A map $S^2 \rightarrow S^2$ mapping a circle packing $P$ to itself induces an automorphism of its carrier (as defined in the preliminaries), which by a slight abuse of language we call an automorphism of the circle packing $P$.}

\begin{proposition}
  \changed{Every triangulation of $\R P^2$ is the carrier of a circle packing in $\R P^2$.}
\end{proposition}

\changed{In particular, for any triangulation $T$ of $\R P^2$, this provides a shortest path embedding of $T$. Since any simple graph embedded on $\R P^2$ can be extended to a triangulation, this proves Theorem~\ref{t:shortest_paths} for $\R P^2$.}

\begin{proof}
  \changed{
Let $T$ be a triangulation of $\R P^2$. Let $\pi\colon S^2 \to \R P^2$ be the projection map sending each pair of antipodal points in $S^2$ to a point in $\R P^2$. Let $\hat{T}$ be the double cover of $T$, which is a triangulation of $S^2$ induced by $\pi^{-1}(T)$, and let $i$ be the automorphism $i \colon \hat T\to \hat T$ induced by the antipodal map. By the Koebe--Andreev--Thurston theorem there exists a circle packing $\hat P \subset S^2$ whose carrier is isomorphic to $\hat T$. Furthermore, this circle packing is unique up to M\"{o}bius transformations~\cite[Chapter~13]{t-gt3m-79}, in particular, any automorphism of $\hat{P}$ is induced by a M\"{o}bius transformation of $S^2$. Thus the map $i$ is induced by a M\"{o}bius transformation $\phi\colon S^2 \to S^2$. Furthermore, since $i$ is fixed-point free, so is $\phi$: if $\phi(x)=x$ for $x$ within a circle corresponding to a vertex $v$, then $i(v)=v$ which is a contradiction, and similarly, a fixed point on the intersection of two circles, or in the region between three adjacent circles would correspond to a fixed edge or face for $i$, which is also a contradiction. Now, any fixed point free M\"{o}bius transformation of the sphere is the antipodal map upto a M\"{o}bius transformation~\cite{wil-81}. Specifically, there exists another M\"{o}bius transformation $\tau\colon S^2 \to S^2$ such that $\tau^{-1} \circ \phi \circ \tau$ is the antipodal map. We can conclude that the circle packing $\hat Q =
\tau^{-1}(\hat P)$ is centrally symmetric and therefore it projects to a circle
packing $Q = \pi(\hat Q)$ in $\RP^2$. Since the carrier of $\hat Q$ is
isomorphic to $\hat T$, the carrier of $Q$ is isomorphic to $T$.}
\end{proof}

\subparagraph*{Minimal triangulations.} Let $S$ be a surface and $T$ be a triangulation of it. The triangulation $T$ is
called \emph{reducible}, if it contains an edge $e$ such that the
contraction of $e$ yields again a triangulation, which we denote by
$T/e$. We refer to $e$ as a \emph{contractible} edge (we do not mean
contractibility in a~topological sense). On the other hand, a
triangulation is \emph{minimal} (or \emph{irreducible}), if no edge
can be contracted this way. For every surface there is a finite list
of minimal triangulations. In particular, for the torus $T^2$ this
list consists of $21$ triangulations found by
Lawrencenko~\cite{l-itt-87} and for the Klein bottle $K$ there are
$29$ minimal triangulations found by Sulanke~\cite{s-nitkb-06}.

The strategy of the proof of Theorem~\ref{t:shortest_paths} for $T^2$ and $K$ is to show that it
is sufficient to check Theorem~\ref{t:shortest_paths} for minimal
triangulations with appropriate fixed metric; see Lemma~\ref{l:contraction}.
Then, since every embedded graph can be extended to a triangulation (possibly
with adding new vertices), we finish
the proof by providing the list of shortest path embeddings of the minimal
triangulations.

\begin{lemma}
\label{l:contraction}
Let $S$ be a surface equipped with a flat metric. Let $T$ be a reducible
triangulation with contractible edge $e$. Let us assume that $T/e$ admits a
shortest path embedding with uniqueness into $S$. Then $T$ admits a shortest path
embedding with uniqueness into $S$ as well.
\end{lemma}

The restriction on flat metrics in the lemma above does not seem essential, 
but this is all we need and this way the proof is quite simple.

\bigskip

\begin{proof}

\changed{
Let $v$ be the vertex of $T/e$ obtained by the contraction of $e$. We first consider the shortest path drawing of $T/e$. Then we
perform the appropriate vertex splitting of $v$ (the inverse operation
of the contraction) in a close neighborhood of $v$ so that we get a
shortest path embedding of $T$. In order to see that this is indeed
possible, let us consider the subgraph $G$ formed by the edges incident to
$v$.
It is a simply connected set, which lifts isometrically to the universal
cover so that the edges are realized by straight segments (since
they are shortest paths). Thus we may choose $\varepsilon>0$ small enough such that
the $\varepsilon$-neighborhood $N_G^{\varepsilon}$ of $G$ is simply connected. Moreover, by compactness, for each edge $uv$ of
$T/e$, there exists $\varepsilon' \leq \varepsilon$
such that for every $v'$ in the $\varepsilon'$-neighborhood of $v$, the
geodesic segment connecting $u$ and $v'$ inside
$N_G^{\varepsilon'}$ is the unique shortest path between $u$
and $v'$ in $S$; see Figure~\ref{f:vertex_splitting} and footnote}\footnote{\changed{Indeed, let us consider two
    functions $d, d_{\OUT} \colon S \to \R$. We set $d(x)
    := \dist(u,x)$ and $d_{OUT}(x) = \min\{\dist(u,y) +
    \dist(y,x)\colon y \in S \setminus N_G^\varepsilon\}$. The
    function $d_{OUT}$ is well defined as the function
    $g(y) := \dist(u,y) + \dist(y,x)$ is continuous and attains its minimum on the
    compact set $S \setminus N_G^\varepsilon$. By the triangle
    inequality $|d_{OUT}(x) - d_{OUT}(x')| \leq
    \dist(x,x')$ for $x, x' \in S$ which implies that $d_{OUT}$ is
    continuous. Finally, we observe that $d(v) < d_{\OUT}(v)$
    as the shortest path connecting $u$ and $v$ is unique. Therefore
    there is an open $\varepsilon'$-neighborhood $N_v$ of $v$ inside
    $N_G^\varepsilon$ such that $d(v') < d_{\OUT}(v')$ for any $v'$ in
    $N_v$. This is the required $\varepsilon'$ needed for the edge $uv$.}}. 
    Therefore, it is sufficient to perform the vertex splitting of $v$ in a sufficiently small
  neighborhood of $v$ so that we do not introduce new intersections.  \end{proof}

\begin{figure}
\begin{center}
  \includegraphics{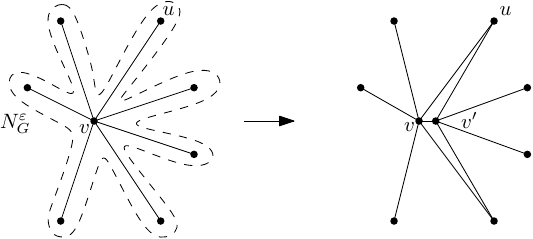}
  \caption{Splitting a vertex in the proof of Lemma~\ref{l:contraction}.}
  \label{f:vertex_splitting}
\end{center}
\end{figure}

%






\subparagraph*{The minimal triangulations of $T^2$ and $K$.} In Figure~\ref{f:torus} we provide a list of shortest path embeddings
with uniqueness of minimal triangulations of the torus with a flat
metric obtained by identifying the opposite edges of the unit
square. They are in the same order as in the book of Mohar and
Thomassen~\cite[Figure~5.3]{mt-gs-01}. The black (thin) edges are the
edges of the triangulation whereas the green (thick) edges are the
identified boundaries of the unit square which are not parts of the edges of the
triangulations. We just skip drawings of the triangulations 7 to 17,
because they are all analogous to the triangulation 6, they only have
different patterns of diagonals. It is clear that every edge is a
geodesic.  In order to check that each of them is drawn as a shortest
path, it is sufficient to verify that each edge projects vertically
and horizontally to a segment of length less than $\frac12$.


\bigskip

For the Klein bottle $K$, we also provide a metric such that all the minimal
triangulations admit shortest path embeddings with uniqueness.
We obtain this metric as the identification of the edges of the
rectangle $R = [0,a]\times[0,b]$, where $a = 1$ and $b = \sqrt{4/3} +
\varepsilon$ for sufficiently small $\varepsilon$. The edges of length
$1$ are identified in coherent directions. The edges of length $b$ are
identified in the opposite directions. The value $b = \sqrt{4/3} +
\varepsilon$ is set up in such a way that if we consider the points $p
= (0, \frac34b) = (1,\frac b4)$ and $q = (\frac13,\frac b4)$ of $K$,
then the shortest path between $p$ and $q$ is the horizontal path of
height $\frac 14b$. However, when we shift $p$ along the boundary of
$R$ a little bit closer to the center, say by $\frac 1{1000}$, then
the shortest path becomes the diagonal edge connecting the left copy
of $p$ and $q$, see Figure~\ref{f:shortest}.

\begin{figure}
\begin{center}
  \includegraphics{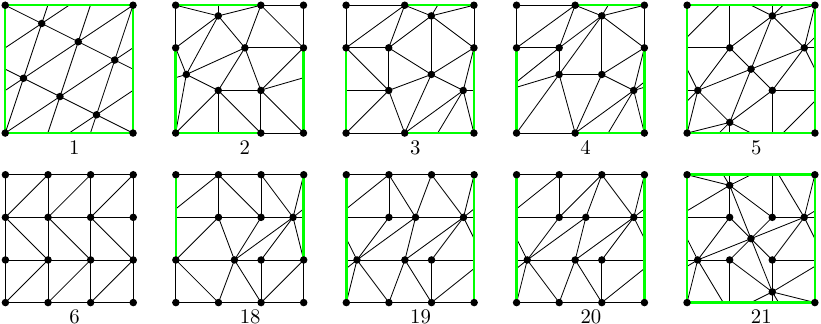}
  \caption{Minimal triangulations of the torus.}
  \label{f:torus}
\end{center}
\end{figure}

\begin{figure}
\begin{center}
  \includegraphics{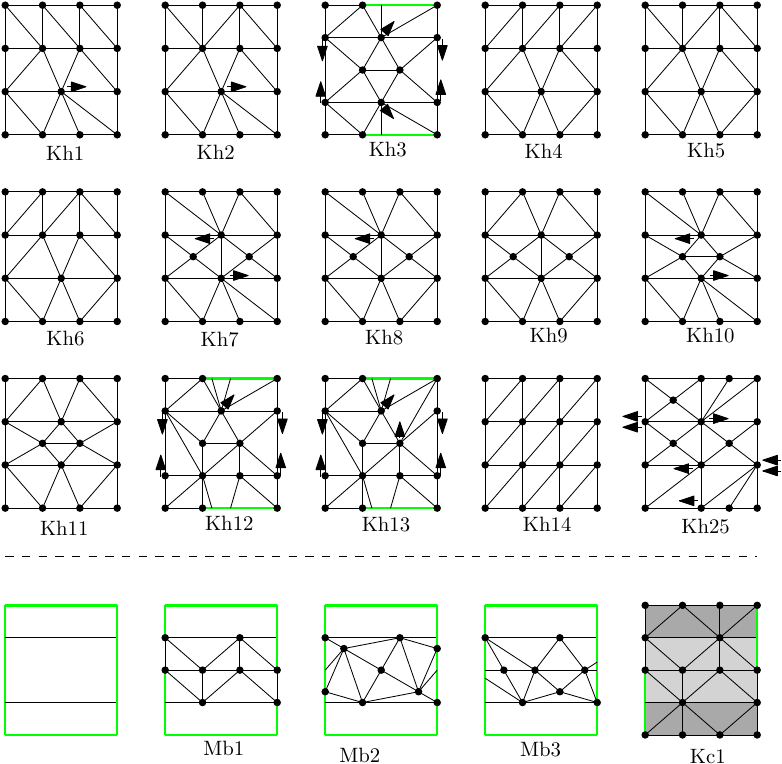}
  \caption{Minimal triangulations of the Klein bottle \changed{(for Kc1 we indicate the
  two copies of Mb1 by different shades of grey).}}
\label{f:all-cases}
\end{center}
\end{figure}

\begin{figure}
\begin{center}
\includegraphics{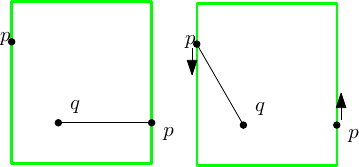}
\caption{Shortest paths in the Klein bottle.}
\label{f:shortest}
\end{center}
\end{figure}


There are 29 minimal triangulations of the Klein bottle. A list of 25
of them was first found by Lawrencenko and
Negami~\cite{ln-itkb-97}. Later on, Sulanke~\cite{s-nitkb-06} found a
gap in the claimed completeness of this list and provided a complete
list containing 4 additional triangulations. These triangulations
split into two classes. The \changed{25} triangulations of the first class are
named Kh1--Kh25 and \changed{the four triangulations of the second class are named Kc1--Kc4.} 
The triangulations from the second class \changed{are those that} contain a cycle of length $3$ which splits the Klein bottle into two M\"{o}bius bands.

We \changed{begin examining} the triangulations \changed{of} the first
class. We present shortest path embeddings with uniqueness for 15 of
them; see the top three lines of Figure~\ref{f:all-cases}. We omit the
triangulations Kh15--Kh24 \changed{because they are very similar} to Kh14, only the
diagonal edges form a different pattern. 
The vertices of the triangulations are positioned in
lattice points of the lattice generated by vectors $(\frac1{12}, 0)$ and
$(0,\frac b{12})$. In some cases an additional
shift is necessary by a small value $\frac1{1000}$ (but this value is
large compared to $\varepsilon$): this is indicated by arrows next to
the vertices. (The pair of arrows in Kh25 indicates a shift by
$\frac{2}{1000}$.)
Most of the drawings are very similar to the drawings by
Negami, Lawrencenko, and Sulanke. Only for the drawings of Kh3,
Kh12, Kh13, and Kh25 we did more significant movements. It is routine (but
tedious) to check that all the edges are indeed drawn as shortest paths. For
many edges this can be checked easily. For few not so obvious cases the
general recipe is to use the universal cover approach and
Lemma~\ref{l:VoronoiCell}. 

Now let us focus on the triangulations in the second class. All of them are
obtained by gluing two triangulations of the M\"{o}bius bands along their
boundaries. In our case, we split $K$ into two bands by a cycle depicted on the
bottom left picture of Figure~\ref{f:all-cases}. There is an isometric homeomorphism which maps one
band to another and which preserves the common boundary pointwise. Therefore,
it is sufficient to present the shortest paths embeddings with uniqueness into the 
bands, as on the middle three pictures. Then we get drawings of Kc1--Kc4 using this
homeomorphism. For example, Kc1 is obtained by gluing two copies of Mb1
together\changed{, as depicted on the bottom right picture of
Figure~\ref{f:all-cases}}. The vertices on the pictures are the lattice points of the same
lattice as above with exception of two points of Mb3. The points on the
`central' cycle of Mb3 have coordinates $(\frac16, \frac b2)$, $(\frac 49,
\frac b2)$, and $(\frac 89, \frac b2)$. Note that we have \changed{significantly} redrawn the original
drawings of Lawrencenko and Negami~\cite{ln-itkb-97}, but it is
easy to check that we get the same triangulations, because \changed{each triangulation of the M\"{o}bius band Mb1--Mb3} is quite small.

\section{Square flat metric on the Klein bottle}\label{S:klein}

The task of this section is to prove the following theorem.

\tklein*

We consider the minimal triangulation Kc1 (see
Figure~\ref{f:all-cases}\changed{, bottom, right}) and we denote by $G$ the underlying graph
for this triangulation. We will prove that $G$ does not admit a
shortest path embedding into $K$ with the square metric.  First, we
observe that the triangulation Kc1 is the only embedding of $G$ into
$K$.

\begin{proposition}
\label{p:unique_Kc1}
$G$ has a unique embedding into the Klein bottle.
 \end{proposition}

 \begin{proof}[Proof of Proposition~\ref{p:unique_Kc1}]
$G$ has $9$ vertices and their degrees
are $(8,8,8,5,5,5,5,5,5)$. Since no other irreducible triangulation of
the Klein bottle has this degree sequence, any other hypothetical
embedding of $G$ into the Klein bottle is either non-cellular or has a
reducible edge. In the first case, it means that $G$ is cellularly
embeddable into the sphere or the projective plane, which is not the
case. Indeed, it is obtained as the gluing of two copies of $K_6$
along a triangle, and therefore contains $K_5 \oplus K_5$ (two copies
of $K_5$ identified along an edge minus that edge) as a minor, which
does not embed into the projective plane~\cite[Figure~6.4]{mt-gs-01}. In
the latter case, we observe that an edge contraction cannot decrease
the degree of all three degree $8$ vertices, and thus we reach a
contradiction since a triangulation on $8$ vertices cannot have a
degree $8$ vertex.
\end{proof}


For contradiction, let us assume that $G$ admits a shortest paths embedding into
$K$. We know that Kc1 is obtained by gluing two triangulations of a M\"{o}bius
band along a cycle of length $3$ (the triangle corresponding to this cycle is
not part of the triangulation). Let $abc$ be this cycle. With a slight abuse of
notation we identify this cycle with its image in the (hypothetical) shortest
path embedding into $K$. Our strategy is to show that already $abc$ cannot be embedded into
$K$ with shortest path edges, which will give the required contradiction.
By Proposition~\ref{p:unique_Kc1}, we know that $abc$ splits $K$ to two
M\"{o}bius bands.



Let $X = \R^2$ be the universal
cover of $K$ (with standard Euclidean metric). 
Let $\pi\colon X \to K$ be the isometric projection corresponding to the cover. 
We will represent the Klein bottle with the flat-square metric as the unit square 
$[0,1]^2$ with suitable identification of the edges ($aba^{-1}b$, as in the previous
section). We will use the convention that
$\pi((0,1)^2) = (0,1)^2$; that is, the projection is the identity on the
interior of this square. See Figure~\ref{f:uc}.

Given a point $p \in K$ we set
$X_p := \pi^{-1}(p)$. Finally, let $\V_p$ be the Voronoi diagram in $X$
corresponding to the set $X_p$.   

\begin{lemma}
\label{l:VoronoiCell}
  Let $p$ and $q$ be two points in $K$ and $\gamma$ be an arc (edge) connecting
  them, considered as a subset of $K$. Then $\gamma$ is the unique shortest path 
  between $p$ and $q$ if and only if there are $p' \in X_p$, $q' \in X_q$ such
  that $\gamma = \pi(\overline{p'q'})$ where $\overline{p'q'}$ denotes the straight edge
  connecting $p'$ and $q'$ in $X$ and $q'$ belongs to the open Voronoi cell for $p'$ in
  $\V_p$.
\end{lemma}

\begin{proof}[Proof of Lemma~\ref{l:VoronoiCell}]
  Any path $\kappa$ with endpoints $p$ and $q$ lifts to some path $\kappa'$ with
  endpoints $p' \in X_p$ and $q' \in X_q$ ($\kappa'$, $p'$, and $q'$ are not
  determined uniquely). This lift preserves the length of the path. Vice versa, any path 
  connecting a point in $X_p$ with a point in $X_q$ projects to a path
  connecting $p$ and $q$ (not necessarily simple), again preserving the length.

  Therefore, $\gamma$ is the shortest path in $K$ connecting $p$ and $q$ if and only if it
  lifts to a straight edge realizing the distance between $X_p$ and $X_q$ in $X$. 
  Such an edge connects $p' \in X_p$ and $q' \in X_q$.
  By symmetry, we can fix $q'$ arbitrarily and we look for the closest $p'$. Then,
  a point $p'$ is the unique point of $X_p$ closest to $q'$ if and only if $q'$ 
  belongs to the open Voronoi cell for $p'$ in $\V_p$. This is what we need.
\end{proof}


Now let us lift the cycle $abc$ to a path $a'b'c'a''$
in $X$; see Figure~\ref{f:uc}. 
Given a curve in $X$, we call the length of its projection to the $x$-axis, the
``horizontal length'' of the curve; similarly we speak about the horizontal distance and the
vertical distance of two points in $X$.

\begin{figure}
\begin{center}
  \includegraphics{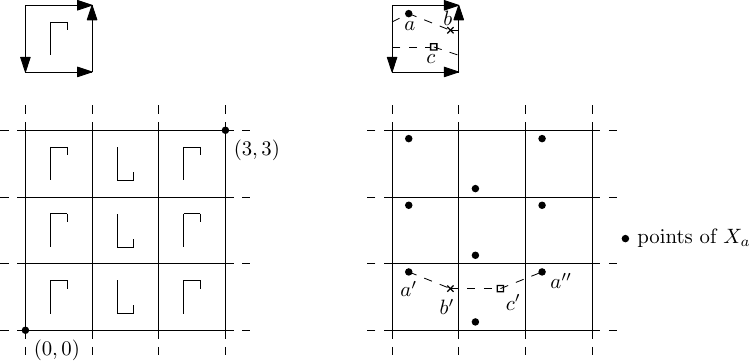}
  \caption{The Klein bottle with a letter `$\Gamma$', and its universal cover
  (left). A lift of the cycle $abc$ (right).}
  \label{f:uc}
\end{center}
\end{figure}

\begin{lemma}
\label{l:a'a''}
The horizontal distance between $a'$ and $a''$ is at least $2$.
\end{lemma}

\begin{proof}[Proof of Lemma~\ref{l:a'a''}] 
If we consider the point $a'$ fixed, then the position of $a''$ in $X_a$ determines the
homotopy class of the cycle $abc$ in $\pi_1(K)$. Therefore, it also determines
the homology classes of this cycle in $H_1(K; \Z_2)$ and in $H_1(K; \Z)$. We
note that the cycle $abc$ must be homologically trivial in $H_1(K; \Z_2)$
because it bounds a M\"{o}bius band; however, it is homologically nontrivial in
$H_1(K; \Z)$ because it bounds a M\"{o}bius band (which is non-orientable) 
on both sides. In addition the cycle $abc$ is two sided, that is, its (regular)
neighborhood is an annulus and not a M\"{o}bius band.
  
The horizontal distance between $a'$ and $a''$ must be a non-negative integer.
We will rule out the cases when this distance is $0$ or $1$. 

If this distance is $1$, then the cycle $abc$ is not two-sided (this can be
read on the lift), a contradiction.

If the horizontal distance is $0$ and the vertical distance is odd, then $abc$
is homologically nontrivial in $H_1(K; \Z_2)$. (It is
sufficient to consider the segment connecting $a'$ and $a''$ and project it to a
cycle $z$ in $K$. Then $z$ is homotopy equivalent to $abc$.) 
A contradiction.

Similarly, if the horizontal distance is $0$ and the vertical distance is even,
then $abc$ is homologically trivial in $H_1(K; \Z)$. (Again we project the
segment connecting $a'$ and $a''$.) A contradiction.
\end{proof}

\begin{lemma}
\label{l:5/8}
Let $\gamma$ be a unique shortest path in $K$ connecting points $p$ and $q$. Let
$\gamma'$ be a lift of $\gamma$ with endpoints $p'$ and $q'$. Then the
horizontal distance in $X$ between $p'$ and $q'$ is less than $\frac58$.
\end{lemma}

\begin{proof}

  \medskip
  
\renewcommand*{\windowpagestuff}{\includegraphics{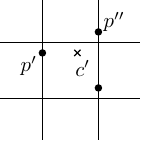}}
\opencutleft
  \begin{cutout}{2}{0pt}{11cm}{6}
  Let $C$ be the open Voronoi cell for $p'$ in $\V_p$. By
  Lemma~\ref{l:VoronoiCell}, $q'$ belongs to $C$. Therefore, it is sufficient
  to check that every point $c'$ of $C$ has horizontal distance less than $\frac58$
  from $p'$. Without loss of generality, we may assume that the $x$-coordinate
  of $p'$ equals $0$ since shifting $p'$ in horizontal direction only shifts
  $X_p$ and $\V_p$ (note that this is not true for the vertical direction). 
  For contradiction, there is a $c'$ in $C$ at distance at least $\frac
  58$ and without loss of generality the $x$-coordinate of $c'$ is positive.
  Let $p''$ be the point of $X_p$ with $x$-coordinate equal $1$ which is
  vertically closest to $c'$ (pick any suitable point in case of draw); see the picture on the left. The vertical distance between $c'$ and $p''$ is
  at most $\frac12$. A simple calculation, using the Pythagoras theorem, gives
  that $p''$ is at most as far from $c'$ as $p'$. A contradiction.
\end{cutout}
\end{proof}
 Finally, we summarize how the previous lemmas yield a contradiction.
%
 By Lemma~\ref{l:a'a''}, the horizontal distance between $a'$ and $a''$ is at
 least $2$. On the other hand, Lemma~\ref{l:5/8} gives that the horizontal 
 length of each of the edges $a'b'$, $b'c'$, and $c'a''$ is at most $\frac58$, 
 altogether at most $\frac{15}8$. This gives the required contradiction, which
 finishes the proof of Theorem~\ref{t:klein}.



\section{Asymptotically almost all hyperbolic metrics are not universal}\label{S:pants}

Before stating the main theorem of this section, we will give some very quick
background on the geometry of surfaces, we refer to Farb and
Margalit~\cite{fm-pmcg-11} for a proper introduction. The Teichm\"uller
space $\mathcal{T}_g$ of a surface $S$ of genus $g$ denotes the set of
hyperbolic metrics on $S$, such that two metrics are equivalent if
they are related by an isometry isotopic to the identity. In some
contexts, like ours, one might also want to identify metrics related
by an isometry (not necessarily isotopic to the identity). The
corresponding space is called the \textit{moduli space} $\mathcal{M}_g$ of the
surface, and is obtained by quotienting $\mathcal{T}_g$ by the
\textit{mapping class group} of $S$, i.e., its group of
homeomorphisms. This moduli space can be endowed with multiple
structures, here we will be interested in a particular one, called the
Weil--Petersson metric. This metric provides $\mathcal{M}_g$ with a
Riemannian structure of finite volume, and therefore by renormalizing,
we obtain a probability space, allowing to choose a random metric. We
can now state the main theorem of this section.

\Tnegative*

The proof is a consequence of two important results on random
hyperbolic metrics. The first is a small variant of a theorem of Guth,
Parlier, and Young~\cite[Theorem~1]{gpy-pdrs-11} that relies on the work of Wolpert~\cite{w-gwpcts-03}. Before stating it, we
need some definitions.

Given a hyperbolic metric $m$ on a surface
$S$, we say that $m$ has total pants length at least $\ell$ if in any
pants decomposition $\Gamma$ of $S$, the lengths of the closed curves of
$\Gamma$ sum up to at least $\ell$. We say that $m$ has total pants
length of type $\xi$ at least $\ell$ if in any pants decomposition $\Gamma$ of
$S$ of type $\xi$, the lengths of the closed curves of $\Gamma$ sum up
to at least $\ell$.

\begin{theorem}\label{T:gpy}
  For any $\varepsilon >0$ and any family of types of pants decomposition
  \changed{$(\xi_g)$}, a random metric on $\mathcal{M}_g$ has total pants length
  \changed{of type} $\xi_g$ at least $g^{4/3 - \varepsilon}$ with probability tending to $1$ as $g \rightarrow \infty$.
\end{theorem}

\begin{proof}[Proof of Theorem~\ref{T:gpy}]
This bound is obtained with a similar technique as the proof of Theorem~1 of Guth, Parlier, and Young~\cite{gpy-pdrs-11}. We refer to their article for more details, and as in their proof, we will discard non super-exponential terms, e.g., $n! \approx n^n$.
For every $a,b,c \in \mathbb R^+$ there exists a unique hyperbolic metric on a pair of pants with boundary lengths $a,b$ and $c$ \changed{(see for example Ratcliffe~\cite[Theorem~9.7.3]{r-fhm-06})}. For a pants 
  decomposition of fixed type $\xi_g$, the Weil--Petersson volume form on moduli space is the push forward of the form
   $d\ell_1 \wedge \ldots \wedge d\ell_{3g-3} \wedge d\tau_1
  \wedge \ldots \wedge d\tau_{3g-3}$ on Teichm\"uller space which is identified with $\mathbb R^{6g-6}$ and the $\ell_i$ denote the
  lengths of the (geodesic) boundaries of the pants decomposition, while
  the $\tau_i$ quantify how much the metric \textit{twists} around
  each geodesic. Since every full twist gives a homeomorphic metric,
 the subset of Teichm\"uller space $\{(\ell_i, \tau_i) \mid \sum_i \ell_i \leq L, 0 \leq \tau_i
  \leq \ell_i\}$ projects surjectively onto the region of moduli space corresponding to 
  surfaces with total pants length of type $\xi_g$ at most $L$. The volume of this set is bounded by
$\lesssim (\frac{L}{g})^{6g}$, which is to be compared with the total volume of moduli space $\approx g^{2g}$. For $L$ smaller than $g^{4/3-\varepsilon}$, the ratio tends to zero, which proves the theorem.  
\end{proof}

%

The following is an immediate corollary of this theorem.

\begin{corollary}\label{map-pants}
Let $T_g$ be a family of triangulations of $S_g$, such that every member of $T_g$ contains a pants decomposition of fixed type $\xi_g$. For any $\varepsilon >0$, with probability tending to $1$ as $g \rightarrow \infty$, a shortest embedding of $T_g$ into a random hyperbolic surface of genus $g$ has length at least $\Omega(g^{4/3-\varepsilon})$.
\end{corollary}


The next theorem was proved by Mirzakhani~\cite[Theorem~4.10]{m-gwpvrh-13}.
\begin{theorem}\label{Mirzakhani}
With probability tending to $1$, the diameter of a random hyperbolic surface of genus g is $O(\log g)$.
\end{theorem}

Theorem~\ref{negative} is proved by providing an explicit family of
graphs $G_g$ which will embed badly. It is
defined in the following way for $g \geq 2$. Let $\xi_g$ be a type of
pants decompositions for every value of $g$.
\begin{itemize}
\item We start with a pants decomposition of type $\xi_g$ of a surface $S_g$.
\item We place four vertices on every boundary curve.
\item We triangulate each pair of pants with a bounded size triangulation 
  so that each cycle of length $3$ bounds a triangle in the triangulation, and any
  path connecting two boundary components of the pair of pants has length at
  least $4$ (in particular $G_g$ is a simple graph and each cycle of length $3$
  in the graph $G_g$ bounds a triangle in the triangulation).  
\end{itemize}

The following proposition controls the issues related to the flexibility of embeddings of graphs into surfaces.

\begin{proposition}\label{P:embeddings}
There is a unique embedding of $G_g$ into $S_g$, up to a homeomorphism; in
particular every embedding contains a pants decomposition of type $\xi_g$.
\end{proposition}

\begin{proof}[Proof of Proposition~\ref{P:embeddings}]
Let $v$ be the number of vertices, $e$ be the number of edges and $t$ be the
number of triangles of the triangulation in the definition of $G_g$ (triangles in the graph-theoretical sense). By Euler's formula and by the construction we get $v - e + t = \chi$ where
$\chi$ is the Euler characteristic of $S_g$. Let us consider an embedding
$\Psi$ 
of $G_g$ into $S_g$. Let $f$ be the number of
faces of this embedding and $F$ be the set of faces. Euler's formula for this
embedding gives $v - e + f \geq \chi$ (we get an inequality because some of the
faces need not be embedded cellularly). In particular, we get $f \geq t$. On the
other hand, we get $2e = 3t$ and $2e = \sum_{\sigma \in F} \deg \sigma \geq 3f$
since each edge is in exactly two faces. This gives $3t \geq 3f$. Therefore,
both of the aforementioned inequalities have to be equalities. In particular,
each $\sigma \in F$ is a triangle bounded by a cycle of length $3$ in $G_g$. 
Since the number of cycles of length $3$ in $G_g$ equals $t = f$, we deduce
that $\Psi$ coincides with the embedding from the definition of $G_g$ up to a
homeomorphism.
\end{proof}

\textbf{Remark:} We preferred to use a hands-on construction of the graphs $G_g$, but another approach could be to rely on the theory of LEW-embeddings and use one of its results on uniqueness of embeddings, see for example Mohar and Thomassen~\cite[Corollary~5.2.3]{mt-gs-01}.

With these three results at hand we are ready to provide a proof of the theorem.

\begin{proof}[Proof of Theorem~\ref{negative}]

We use the family of graphs $G_g$ previously defined. Since there are
$O(g)$ curves in a pants decomposition, it contains $O(g)$ edges, and
every embedding of $G_g$ into $S_g$ contains a pants decomposition of type $\xi_g$ by Proposition~\ref{P:embeddings}.

Now, by Corollary~\ref{map-pants}, for every $\varepsilon>0$, and for
$g$ large enough, the probability that the shortest possible embedding
of $G_g$ into a random metric has length at least
$O(g^{4/3-\varepsilon})$ is at least $1 - \varepsilon/2$. In
particular, since there are $O(g)$ edges in $G_g$, some edge $e_g$ in
this embedding must have length at least
$\Omega(g^{1/3-\varepsilon})$. By Theorem~\ref{Mirzakhani}, we can
choose $g$ large enough so that with probability at least
$1-\frac{\varepsilon}{2}$, the random hyperbolic metric has diameter
$O(\log g)$. Hence, by the union bound, with probability
$1-\varepsilon$ both properties hold. Therefore, for every
$\varepsilon>0$, there exists some value $g_0$ such that for any $g
\geq g_0$, in any embedding of $G_g$, there exists an edge $e_g=(x,y)$
such that $\ell_m(e_g)=\Omega(g^{1/3-\varepsilon})$, but $d_m(x,y)\leq
diam(m)\leq O(\log g)$. This implies that $e$ is not drawn by a
shortest path. Similarly, subdividing each edge $O(g^{1/3-\varepsilon})$ times will run into the same issue. This concludes the proof.\end{proof}

\section{Higher genus: positive results}\label{S:concat}

\Tconcat*


Our approach to prove Theorem~\ref{T:concat} is to cut the surface
$S_g$ with a \textit{hexagonal decomposition} $\Delta$, so that every edge
of $G$ is cut $O(g)$ times by this decomposition $\Delta$. The construction
to do this is a slight modification of the \textit{octagonal
decompositions} provided by \'E. Colin de Verdi\`ere and Erickson~\cite[Theorem~3.1]{ce-tnpcs-10}. Each of the hexagons is then
endowed with a specific hyperbolic metric $m_H$, and pasting these
together yields the hyperbolic metric $m$ on $S_g$. The hyperbolic
metric $m_H$ is chosen so that the hexagons are \textit{convex}, i.e.,
the shortest paths between points of a hexagon stay within this
hexagon. Therefore, there only remains to embed the graph $G$ cut
along $\Delta$, separately in every hexagon with shortest paths. To do
this, we use a variant of a theorem of Y. Colin de
Verdi\`ere~\cite{c-crgts-91} which generalizes Tutte's barycentric
method to metrics of nonpositive curvature.

\subparagraph*{Hexagonal decompositions.} A \textit{hexagonal decomposition}, respectively an \textit{octagonal
  decomposition} of $S_g$ is an arrangement of closed curves on $S_g$
that is homeomorphic to the one pictured in
Figure~\ref{F:octagons}.b., respectively Figure~\ref{F:octagons}.a. In
particular, every vertex has degree four and every face has six sides,
respectively eight sides.

\begin{figure}
  \centering
  \def\svgwidth{\textwidth}
  
	\begingroup%
  \makeatletter%
  \providecommand\color[2][]{%
    \errmessage{(Inkscape) Color is used for the text in Inkscape, but the package 'color.sty' is not loaded}%
    \renewcommand\color[2][]{}%
  }%
  \providecommand\transparent[1]{%
    \errmessage{(Inkscape) Transparency is used (non-zero) for the text in Inkscape, but the package 'transparent.sty' is not loaded}%
    \renewcommand\transparent[1]{}%
  }%
  \providecommand\rotatebox[2]{#2}%
  \ifx\svgwidth\undefined%
    \setlength{\unitlength}{4143.37778447bp}%
    \ifx\svgscale\undefined%
      \relax%
    \else%
      \setlength{\unitlength}{\unitlength * \real{\svgscale}}%
    \fi%
  \else%
    \setlength{\unitlength}{\svgwidth}%
  \fi%
  \global\let\svgwidth\undefined%
  \global\let\svgscale\undefined%
  \makeatother%
  \begin{picture}(1,0.36559304)%
    \put(0,0){\includegraphics[width=\unitlength,page=1]{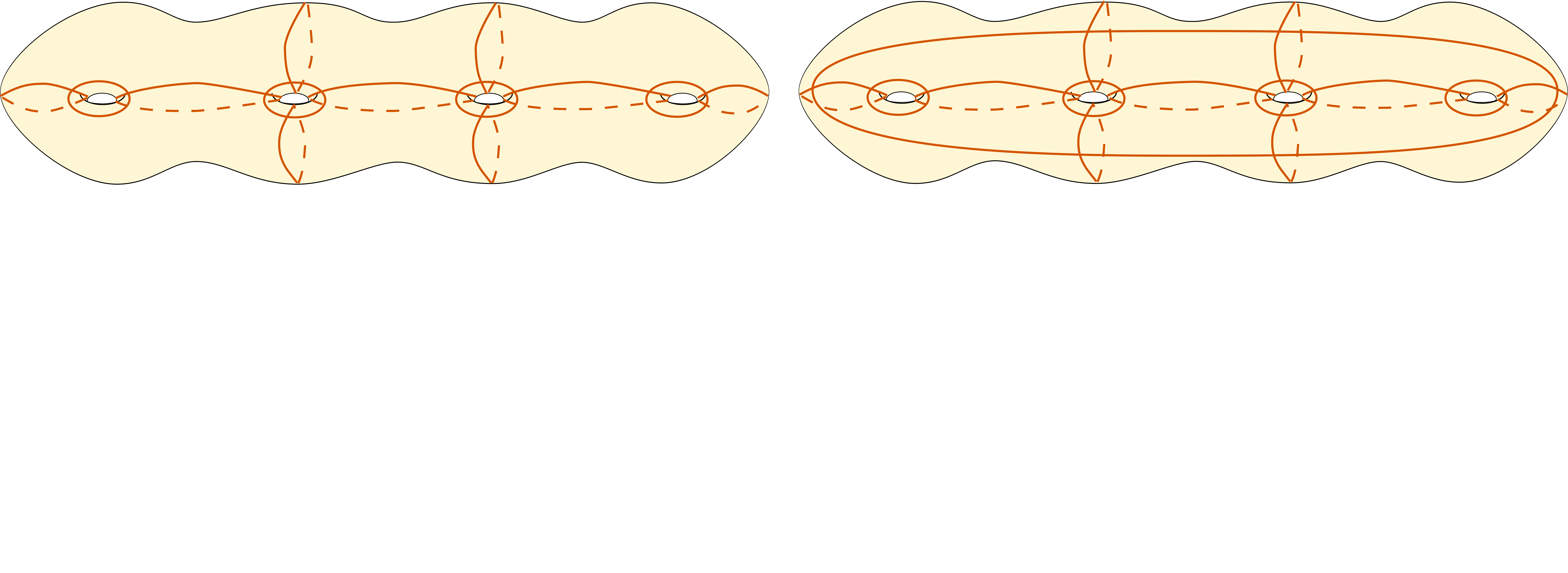}}%
    \put(0.75659616,0.23305333){\color[rgb]{0,0,0}\makebox(0,0)[lb]{\smash{b.}}}%
    \put(0,0){\includegraphics[width=\unitlength,page=2]{octagons.pdf}}%
    \put(0.24335641,0.23531179){\color[rgb]{0,0,0}\makebox(0,0)[lb]{\smash{a.}}}%
    \put(0.50336838,0.00017159){\color[rgb]{0,0,0}\makebox(0,0)[lb]{\smash{c.}}}%
  \end{picture}%
\endgroup%
	
  \caption{a. An octagonal decomposition b. A hexagonal decomposition c. How to add one closed curve to upgrade an octagonal decomposition to a hexagonal decomposition}
\label{F:octagons}
\end{figure}


Octagonal decompositions were introduced by \'E. Colin de Verdi\`ere and Erickson~\cite{ce-tnpcs-10} where they showed how to compute one that does not cross the edges of an embedded graph too many times. We restate their theorem in our language.

\begin{theorem}[{\cite[Theorem~3.1]{ce-tnpcs-10}}]\label{T:octagons}
  Let $G$ be a graph embedded in a surface $S_g$ for $g \geq 2$. There exists an
  octagonal decomposition $\Gamma$ of $S_g$ such that each edge of
  $G$ crosses each closed curve of $\Gamma$ a constant number of times.
 \end{theorem}

We observe that this octagonal decomposition can be upgraded to a hexagonal decomposition that still does not cross $G$ too much:

\begin{corollary}\label{C:hexagons}
  Let $G$ be a graph embedded in a surface $S_g$. There exists a
  hexagonal decomposition $\Delta$ of $S_g$ such that each edge of $G$
  crosses each closed curve of $\Delta$ a constant number of times,
  except for maybe one closed curve which is allowed to cross each
  edge of $G$ at most $O(g)$ times. In particular, the number of crossing
  between every edge of $G$ and $\Delta$ is $O(g)$.
\end{corollary}

%

\begin{proof}
The decomposition $\Delta$ is simply obtained by taking the decomposition $\Gamma$ and adding a single curve that follows closely a concatenation of $O(g)$ subpaths of curves of $\Gamma$, see Figure~\ref{F:octagons}c. The resulting arrangement of curves has the topology of a hexagonal decomposition, and the bounds on the number of crossings results directly from the construction.
\end{proof}

\changed{\textbf{Remark:} We remark that \'E. Colin de Verdi\`ere and Erickson~\cite{ce-tnpcs-10} actually build a decomposition with hexagonal faces before discarding cycles to get an octagonal decomposition. However, their decomposition is not homeomorphic to our notion of hexagonal decomposition, and would be less adapted for the proof of Theorem~\ref{T:concat} since some of their hexagonal faces are glued to themselves (in particular Proposition~\ref{P:convexity} would not hold).}




\subparagraph*{The hyperbolic metric.} We first endow each hexagon of the hexagonal decomposition with the
hyperbolic metric $m_H$ of an equilateral right-angled hyperbolic
hexagon. Since the hexagons have right angles and the vertices of a
hexagonal decomposition have degree $4$, this metric can be safely
pasted between hexagons to endow $S_g$ with a hyperbolic metric
$m$. The main property of this metric that we will use is the
following one:


\begin{proposition}\label{P:convexity}
Every hexagon $H$, viewed as a subset of $S_g$ endowed with $m$, is
convex, i.e., every path between $x,y \in H$ that is a shortest path
in $H$ is also a shortest path in $S_g$.
\end{proposition}

\begin{proof}
\changed{  We will prove that for any two points $x,y \in H$, there exists a
  shortest path (in $S_g$) that is entirely contained in $H$. The
  proof relies on an exchange argument based on the symmetries of the
  hexagonal decomposition. We rely on two involutions which we first
  define on the intersection graph\footnote{The intersection graph of the hexagonal decomposition is defined by taking one vertex for
    each hexagon and edges between adjacent hexagons (we allow
    multiple edges).} of the hexagonal decomposition. This graph should
  be thought of as embedded on the surface of a $g \times 1 \times 1$
  rectangular block (see Figure~\ref{F:intersections}), and
}
\begin{figure}
  \centering
  \def\svgwidth{11cm}
\begingroup%
  \makeatletter%
  \providecommand\color[2][]{%
    \errmessage{(Inkscape) Color is used for the text in Inkscape, but the package 'color.sty' is not loaded}%
    \renewcommand\color[2][]{}%
  }%
  \providecommand\transparent[1]{%
    \errmessage{(Inkscape) Transparency is used (non-zero) for the text in Inkscape, but the package 'transparent.sty' is not loaded}%
    \renewcommand\transparent[1]{}%
  }%
  \providecommand\rotatebox[2]{#2}%
  \ifx\svgwidth\undefined%
    \setlength{\unitlength}{3565.51473278bp}%
    \ifx\svgscale\undefined%
      \relax%
    \else%
      \setlength{\unitlength}{\unitlength * \real{\svgscale}}%
    \fi%
  \else%
    \setlength{\unitlength}{\svgwidth}%
  \fi%
  \global\let\svgwidth\undefined%
  \global\let\svgscale\undefined%
  \makeatother%
  \begin{picture}(1,0.22832924)%
    \put(0,0){\includegraphics[width=\unitlength,page=1]{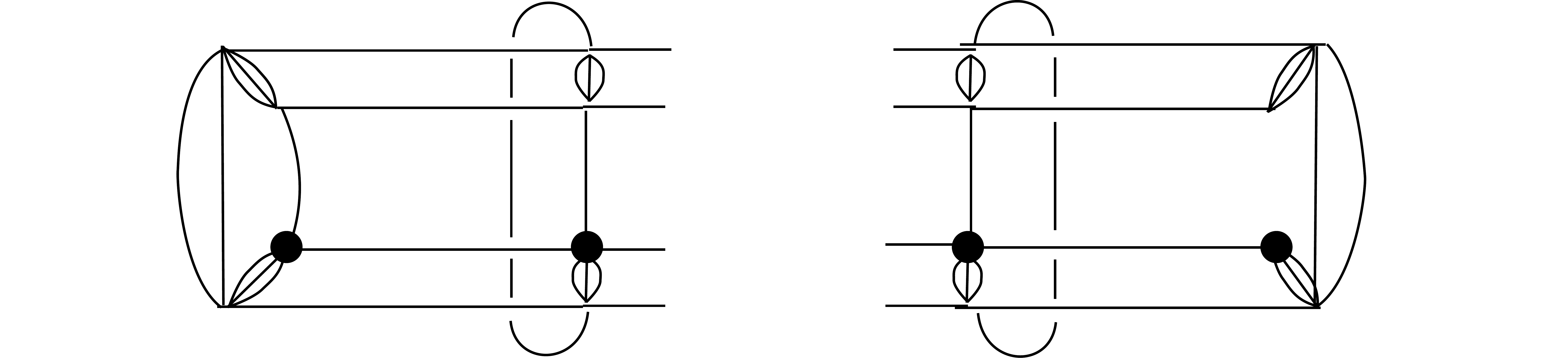}}%
    \put(0.47824722,0.17299713){\color[rgb]{0,0,0}\makebox(0,0)[lb]{\smash{$\ldots$}}}%
    \put(0.47824722,0.04536604){\color[rgb]{0,0,0}\makebox(0,0)[lb]{\smash{$\ldots$}}}%
    \put(0,0){\includegraphics[width=\unitlength,page=2]{Intersection.pdf}}%
    \put(0.9664435,0.11376957){\color[rgb]{0,0,0}\makebox(0,0)[b]{\smash{$\sigma_1$}}}%
    \put(0,0){\includegraphics[width=\unitlength,page=3]{Intersection.pdf}}%
    \put(0.02514539,0.1736382){\color[rgb]{0,0,0}\makebox(0,0)[b]{\smash{$\sigma_2$}}}%
    \put(0.02706857,0.0424907){\color[rgb]{0,0,0}\makebox(0,0)[b]{\smash{$\sigma_2$}}}%
    \put(0,0){\includegraphics[width=\unitlength,page=4]{Intersection.pdf}}%
  \end{picture}%
\endgroup%
  \caption{The intersection graph of the hexagonal decomposition and the two involutions: $\sigma_1$ is the reflection across the horizontal plane pictured in gray, and $\sigma_2$ swaps every pair of triply-linked adjacent vertices (pictured by disks and stars).}
  \label{F:intersections}
\end{figure}

\begin{itemize}
  
 \item \changed{The map $\sigma_1$ is the reflection across the mid-plane of the top and bottom facets, sending vertices above that plane to the corresponding ones below.}
 \item \changed{The map $\sigma_2$ is the reflection accross the mid-plane of the front and back facets. Equivalently, it swaps each hexagon with its neighbor in the original octagonal decomposition.
   }
\end{itemize}
\changed{
Since all the hexagons are isometric, these two maps induce isometric involutions of the surface $S_g$ endowed with $m$, which we also denote by $\sigma_1$ and $\sigma_2$. They allow us to cut $S_g$ into four quadrants, each of them being a linear concatenation of hyperbolic hexagons: the first one, which we denote by $Q_1$ is pictured in Figure~\ref{F:hexagons}, and we obtain $Q_2$, $Q_3$ and $Q_4$ by applying respectively $\sigma_1$, $\sigma_2$ and $\sigma_1 \sigma_2$ to $Q_1$.
}
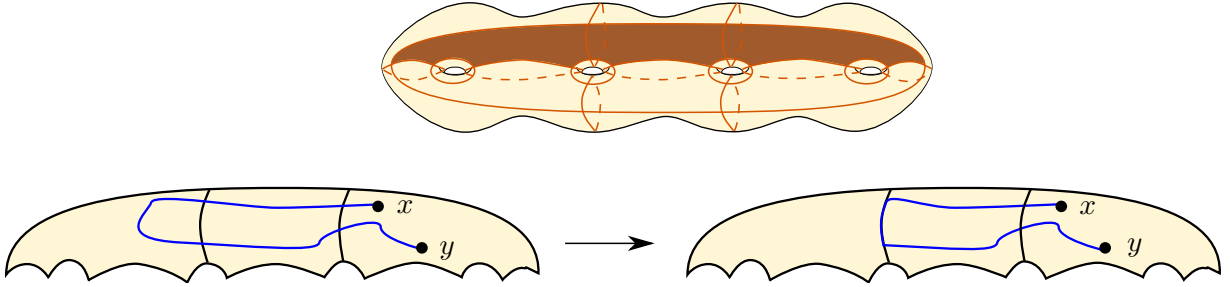
\begin{figure}
  \centering
  \def\svgwidth{16cm}

  \begingroup%
  \makeatletter%
  \providecommand\color[2][]{%
    \errmessage{(Inkscape) Color is used for the text in Inkscape, but the package 'color.sty' is not loaded}%
    \renewcommand\color[2][]{}%
  }%
  \providecommand\transparent[1]{%
    \errmessage{(Inkscape) Transparency is used (non-zero) for the text in Inkscape, but the package 'transparent.sty' is not loaded}%
    \renewcommand\transparent[1]{}%
  }%
  \providecommand\rotatebox[2]{#2}%
  \ifx\svgwidth\undefined%
    \setlength{\unitlength}{4487.02836459bp}%
    \ifx\svgscale\undefined%
      \relax%
    \else%
      \setlength{\unitlength}{\unitlength * \real{\svgscale}}%
    \fi%
  \else%
    \setlength{\unitlength}{\svgwidth}%
  \fi%
  \global\let\svgwidth\undefined%
  \global\let\svgscale\undefined%
  \makeatother%
  \begin{picture}(1,0.23356569)%
    \put(0,0){\includegraphics[width=\unitlength,page=1]{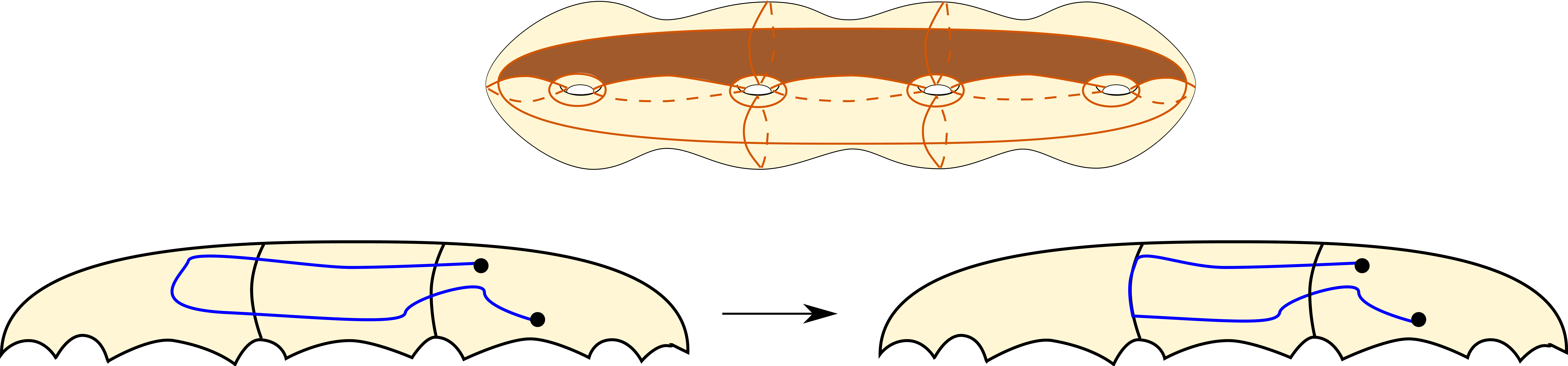}}%
    \put(0.32207872,0.0587886){\color[rgb]{0,0,0}\makebox(0,0)[lb]{\smash{$x$}}}%
    \put(0.35722768,0.02414907){\color[rgb]{0,0,0}\makebox(0,0)[lb]{\smash{$y$}}}%
    \put(0.8835921,0.05803199){\color[rgb]{0,0,0}\makebox(0,0)[lb]{\smash{$x$}}}%
    \put(0.92266715,0.02720549){\color[rgb]{0,0,0}\makebox(0,0)[lb]{\smash{$y$}}}%
  \end{picture}%
\endgroup%
  \caption{Top: The surface is cut into four quadrants, and $Q_1$ is pictured here with a darker color. Bottom: If $\gamma''$ enters another hexagon than $H$, it can be shortcut.}
  \label{F:hexagons}
\end{figure}

\changed{Now, let $x$ and $y$ be two points in a hexagon $H$, which we assume without loss of generality (by applying $\sigma_1$ and/or $\sigma_2$) to be in $Q_1$, and let $\gamma$ be a shortest path between $x$ and $y$. This path $\gamma$ may wander out of $H$, but in this case we show that there is another shortest path between $x$ and $y$ that is entirely contained within $H$. For every maximal subpath $\alpha$ of $\gamma$ in the interior of $Q_3 \cup Q_4$, we reflect $\alpha$ into $Q_1 \cup Q_2$ using $\sigma_2$. This results in a new path $\gamma'$ between $x$ and $y$, with the same length as $\gamma$ (since $\sigma_1$ is an isometry) that is entirely contained in $Q_1 \cup Q_2$. Then, for every maximal subpath $\alpha'$ of $\gamma'$ in the interior of $Q_2$, we reflect $\alpha'$ into $Q_1$ by applying $\sigma_1$, which yields a path $\gamma''$ between $x$ and $y$ entirely contained within $Q_1$. Since $\gamma''$ has the same length as $\gamma$, it is also a shortest path.}

\changed{At this stage, we claim that $\gamma''$ is actually contained in $H$. Indeed, if it were not, since $Q_1$ is a linear concatenation of hexagons, there would be another hexagon $H'$ that contains a subpath $\alpha''$ of $\gamma''$ such that the two endpoints $e_1$ and $e_2$ of $\alpha''$ lie on the same side $s$ of $H'$, and $\alpha'' \not \subset s$. But then $\gamma''$ cannot be a shortest path, since this subpath $\alpha''$ could be shortcut by following the side $s$ between $e_1$ and $e_2$ instead of entering $H'$, see Figure~\ref{F:hexagons}.}

\changed{Thus we have found a shortest path between $x$ and $y$ that is entirely contained within $H$, which proves the proposition.
}
\end{proof}

\subparagraph*{Finishing the proof.} We prove in this paragraph how to reembed a graph embedded in a hexagon so that its edges are shortest paths. This allows us to finish the proof.

\begin{restatable}{theorem}{TYves}\label{T:Yvesvariant}
Let $G$ be a graph embedded as a triangulation in a hyperbolic hexagon
$H$ endowed with the metric $m_H$. If there are no dividing edges in
$G$, i.e., edges between two non-adjacent vertices on the boundary of
$H$, then $G$ can be embedded with geodesics, with the vertices on the
boundary of $H$ in the same positions as in the initial embedding.
\end{restatable}


Let us postpone the proof of this theorem for now, and show how to conclude the proof of Theorem~\ref{T:concat}. We first show how to upgrade a graph embedded in a disk to a triangulation.

\begin{lemma}\label{L:triangulate}
For any graph $G$ embedded in a disk without
dividing edges, there exists a triangulation $G'$ of the disk \changed{that contains}
$G$ as a subgraph and that does not contain any dividing edges.
\end{lemma}

\begin{proof}
For every face $F$ of $G$, we start by adding a vertex in $F$ and
edges connecting it to the vertices adjacent to the face. This does
not add loops or dividing edges, but may add multiple edges if one
vertex occurs multiple times on the boundary of a face. These are
taken care of by subdividing them once again and triangulating.
\end{proof}

All the pieces are now in place for the proof of Theorem~\ref{T:concat}.

\begin{proof}[Proof of Theorem~\ref{T:concat}]
By Corollary~\ref{C:hexagons}, one can embed $G$ into $S_g$ such that
every edge of $G$ is cut $O(g)$ times by the hexagonal decomposition
$\Delta$. This defines a graph $G'= \bigcup_i G'_i$ such that each of
the graphs $G'_i$ is embedded in a single hexagon and $G'$ is obtained
from $G$ by subdividing every edge $O(g)$ times. If there are dividing
edges in $G'_i$, they can be removed by subdividing the edge once. By Lemma~\ref{L:triangulate}, one can upgrade all
the $G'_i$ to triangulations. We can then apply Theorem~\ref{T:Yvesvariant}
in each of the hexagons separately, yielding embeddings with shortest
paths.
Since the vertices on the boundary did not move during the
reembedding, this defines an embedding of $G$ into $S_g$. Since $H$ is
simply connected and $m_H$ is hyperbolic, there is a unique geodesic
connecting any two points, and this geodesic is a shortest path.
Therefore the edges of $G'$ are shortest paths in $H$. By
Proposition~\ref{P:convexity}, each edge of $G'$ is also a shortest
path in $S_g$. Therefore each edge of $G$ is embedded as a
concatenation of $O(g)$ shortest paths.
\end{proof}

We note that by subdividing each edge once more, the shortest paths we
obtain are unique.
\bigskip 

The proof of Theorem~\ref{T:Yvesvariant} is obtained in a spirit
similar to the proof of the one of the celebrated \textit{spring
  theorem} of Tutte~\cite{t-hdg-63}. However, there are two main
differences which prevent us from directly appealing to the literature: on the
one hand the metric is not Euclidean but hyperbolic, and on the other
hand the boundary of the input polygon is not \textit{strictly
  convex}, since there may be multiple vertices of $G$ on a geodesic
boundary of $H$. The hypothesis on dividing edges is tailored to
circumvent the second issue, and in a Euclidean setting it was proved
by Floater~\cite{f-otoplm-03} that the \changed{corresponding} embedding theorem
holds. Regarding the first issue, Y. Colin de Verdi\`ere stated a
Tutte embedding theorem~\cite[Theorem~3]{c-crgts-91} for the
hyperbolic setting with strictly convex boundary, yet he actually did not
provide a proof for it.
\ifappendix 
In Appendix~\ref{A:Yves}
\fi
we show how to prove Theorem~\ref{T:Yvesvariant}
in the generality that we need following the ideas laid out by Y. Colin de Verdi\`ere in the rest of his article~\cite{c-crgts-91}. This concludes the proof of Theorem~\ref{T:concat}.

\bigskip
%

Finally, we remark that this proof technique provides an alternative proof of Negami's Theorem~\cite{n-cngepc-01} for orientable surfaces. If $G_1$ and $G_2$ are two graphs embedded on the orientable surface of genus $g$, a crude application of Theorem~\ref{T:concat} shows that one can reembed both graphs with a homeomorphism such that each edge is realized as a concatenation of $O(g)$ shortest paths for our hyperbolic metric. Since hyperbolic shortest paths in general position cross at most once, this gives embeddings of $G_1$ and $G_2$ such that there are $O(g^2)$ crossings between each edge of $G_1$ and each edge of $G_2$. Negami proved that $O(g)$ crossings are actually enough, and a deeper look at our construction also achieves this better bound: it is easy to see that in our reembeddings, each edge is actually cut into $O(1)$ subedges realized as shortest paths in each hexagon. Since there are $O(g)$ hexagons, there are in total $O(g)$ crossings between each edge of $G_1$ and $G_2$, which yields the following:

\begin{corollary}
\label{c:negami}  
There exists an absolute constant $c>0$ such that if $S_g$ is
an orientable surface of genus $g$, for any two embedded graphs $G_1, G_2 \to
S_g$ there exists a homeomorphism $h:S_g\rightarrow S_g$ such that
$cr(h(G_1),G_2) \leq c g  |E(G_1)|\cdot |E(G_2)|$\end{corollary}


\subparagraph*{Acknowledgements} We are grateful to \'Eric Colin de Verdi\`ere
for his involvement in the early stages of this research. We also thank Sergio
Cabello, Francis Lazarus, Bojan Mohar, Eric Sedgwick, Uli Wagner and the
anonymous referees for helpful remarks, and Xavier Goaoc for organizing the
workshop that led to this work.

\bibliographystyle{plain}

\ifappendix
\appendix

\section{Tutte's embedding theorem in a hyperbolic setting}\label{A:Yves}


In this section, we explain the proof of the subsequent theorem, following the
arguments of Y. Colin de Verdi\`ere~\cite{c-crgts-91}.

\TYves*

As announced, the proof follows from a spring-like construction,
i.e. we think of the edges of the graph $G$ as springs with some
arbitrary stiffness, the vertices which are not on the boundary are
allowed to move and we prove that the equilibrium state for this
physical system is an embedding of the graph.

For an embedding $\varphi: G \rightarrow H$, denote by
$e_{ij}$ the map $[0,1] \rightarrow H$ representing the edge
$(i,j)$. Starting with an embedding $\varphi_0: G \to H$ and given assignments $c_{i,j} : E(G) \to \mathbb{R}^+$, we are interested in the map $\varphi: G \to
H$ minimizing the energy functional

   \[E_{\varphi}=\sum_{(i,j )\in E} \int_0^1 c_{ij} ||e_{ij}'(t)||^2 dm_H\]
with fixed vertices on the boundary of $H$. This is the equilibrium
state of the spring system with the $c_{i,j}$ coefficients specifying
the stiffness of the springs. We claim that $\varphi$ is an embedding
such that the edges are geodesics.

\textbf{Step 1: Existence.} The existence of $\varphi$ follows from
classical compactness considerations, since an Arzel\`a--Ascoli argument proves the compactness of
sets with bounded energy. Then an extremum of $E_{\varphi}$ corresponds to
a $\varphi$ where all the arcs $e_{i,j}$ are geodesics. Furthermore,
every vertex $\varphi(x)$ which is not on the boundary lies in the
strict hyperbolic convex hull of its neighbors which are not mapped to
the same point.

\textbf{Step 2: Curvature considerations.} Since $\varphi_0$ provides
an embedding of $G$ into $H$, $G$ can be seen as a topological subspace
of $H$. The corresponding simplicial complex will be denoted by $X$
(it is of course homeomorphic to $H$) and its set of vertices, edges and triangles by $V$, $E$, and $T$. By extending $\varphi$ separately with
a local homeomorphism in the interior of each non-degenerate triangle,
we can extend it into a map $\Phi: X \rightarrow H$ agreeing
with $\varphi$ on $G$.

Now, the map $\Phi:X \rightarrow H$ provides values for the angles
of the non-degenerate triangles in $X$. For degenerate
triangles, values of the angles are taken arbitrarily so that they sum
to $\pi$ (therefore morally their hyperbolic area is zero). For an interior vertex
$v$, let us define the curvature $K(v)= 2 \pi - \sum_i \alpha^i_v$, where
$\alpha^i_v$ are the angles adjacent to $v$. For a vertex $v$ on the interior
of a geodesic boundary, we define it by $K(v)= \pi - \sum_i \alpha_v^i$,
and on the six vertices of $H$, we take it to be $K(v)=\pi/2 - \sum_i
\alpha_v^i$.

The area of a geodesic hyperbolic triangle is $\pi$ minus the sum of
its angles. Summing over all the triangles of $\Phi(X)$, we obtain
$|T| \pi - \sum_v \sum_i \alpha^i_v = \sum_{t \in T} Area(t)$. With
Euler's formula and double counting, this gives $\sum_{t \in
  T} Area(t)=\pi + \sum_v K(v)$. Since the boundary is fixed, $\Phi$ has degree one and is thus
surjective, therefore the sum of the areas of the triangles is at least the area of the hexagon, which is $\pi$ since it is
right-angled. Therefore $\sum_v K(v) \geq 0$.

\textbf{Step 3: Punctual degeneracies.} In this step we investigate
which subcomplexes of $X$ can be mapped to a single point. We show
that no triangle can be mapped to a single point, and that a set of
edges mapped to a single point forms a path subgraph in $G$.

Let $X_1$ be a maximal connected subcomplex of $X$ which is mapped to
a point $x$ by $\Phi$. This subcomplex has to be simply connected,
otherwise the region inside could be mapped to $x$ as well which would
reduce the value of $E_{\varphi}$. Since the boundary edges are fixed
by $\varphi$, $X_1$ does not contain any edge on the boundary or
triangle adjacent to the boundary.

For every vertex $v$ in $\Phi^{-1}(x)$, $\Phi(v)=\varphi(v)$ lies in
the strict convex hull of its neighbors which are not mapped to $x$,
as was observed in Step 1. Therefore the angles of the non-degenerate
triangles adjacent to $v$ sum up to at least $2 \pi$. Indeed the
angular opening at $\varphi(v)$ has to be at least $\pi$ by the
convexity hypothesis, but \changed{if a map $\mathbb{S}^1 \rightarrow
\mathbb{S}^1$ is not surjective then every point in the image has at least two pre-images},
in which case this angular opening of at least $\pi$ amounts to at
least $2 \pi$ in the sum of angles around $v$. This shows that
$K(x):=\sum_{v \in \Phi^{-1}(x)} K(v)$ is nonpositive. Since the
boundary edges are fixed, we also have $K(v)\leq 0$ for the vertices
on the boundary.

Summing over all the values of $x$, we obtain that $\sum_v K(v) \leq
0$, and thus this sum is zero by the previous paragraph, and each of
the $K(x)$ is also zero.

From that we infer that $X_1$ contains no triangle: if it did, there
would be at least $3$ preimages of $x$ for which the angles of the
adjacent non-degenerate triangles would sum up to at least
$2\pi$. Summing them into $K(x)$ we would obtain a nonzero
value. Similarly, $X_1$ can only be a linear subgraph of
$G$, and every triangle adjacent to a $X_1$ not reduced to a point is degenerate.


\textbf{Step 4: Linear degeneracies.} Now that we showed that
triangles can not be mapped to points, we show that triangles are not
mapped to lines either, or equivalently that edges are not mapped to
points.

Let $X_2$ be a maximal connected subcomplex of $X$ such that the image
of the triangles of $X_2$ by $\varphi$ are degenerate. Let us assume
that $X_2$ is non-empty. Then the image $\Phi(X_2)$ is an arc of a
geodesic of $H$: indeed if there was a broken line in $\Phi(X_2)$,
around the breaking points there would be non-degenerate triangles
adjacent to a $X_1$ not reduced to a point, which is absurd by the
previous paragraph.


If this geodesic is not a boundary geodesic of $H$, two of the points
on the boundary of $X_2$ are mapped to the endpoints of the arc of
geodesic, and all the other vertices have their adjacent edges within
$X_2$ because of the convexity condition. Therefore, there must be two
arcs connecting the two boundary points, as in the top of
Figure~\ref{F:appendix}, which is impossible in the simplicial complex
$X$.

If this geodesic is on the boundary of $H$, then by the same convexity
argument, two vertices of $\partial X$ must map to the endpoints of
this arc of geodesic, and the other vertices have all their edges
within $X_2$. Therefore there is a dividing edge connecting these two
vertices, as in the bottom of Figure~\ref{F:appendix}, which is a contradiction.

\begin{figure}
  \centering
  \def\svgwidth{11cm}
  
	\begingroup%
  \makeatletter%
  \providecommand\color[2][]{%
    \errmessage{(Inkscape) Color is used for the text in Inkscape, but the package 'color.sty' is not loaded}%
    \renewcommand\color[2][]{}%
  }%
  \providecommand\transparent[1]{%
    \errmessage{(Inkscape) Transparency is used (non-zero) for the text in Inkscape, but the package 'transparent.sty' is not loaded}%
    \renewcommand\transparent[1]{}%
  }%
  \providecommand\rotatebox[2]{#2}%
  \ifx\svgwidth\undefined%
    \setlength{\unitlength}{2882.13199368bp}%
    \ifx\svgscale\undefined%
      \relax%
    \else%
      \setlength{\unitlength}{\unitlength * \real{\svgscale}}%
    \fi%
  \else%
    \setlength{\unitlength}{\svgwidth}%
  \fi%
  \global\let\svgwidth\undefined%
  \global\let\svgscale\undefined%
  \makeatother%
  \begin{picture}(1,0.50644432)%
    \put(0,0){\includegraphics[width=\unitlength,page=1]{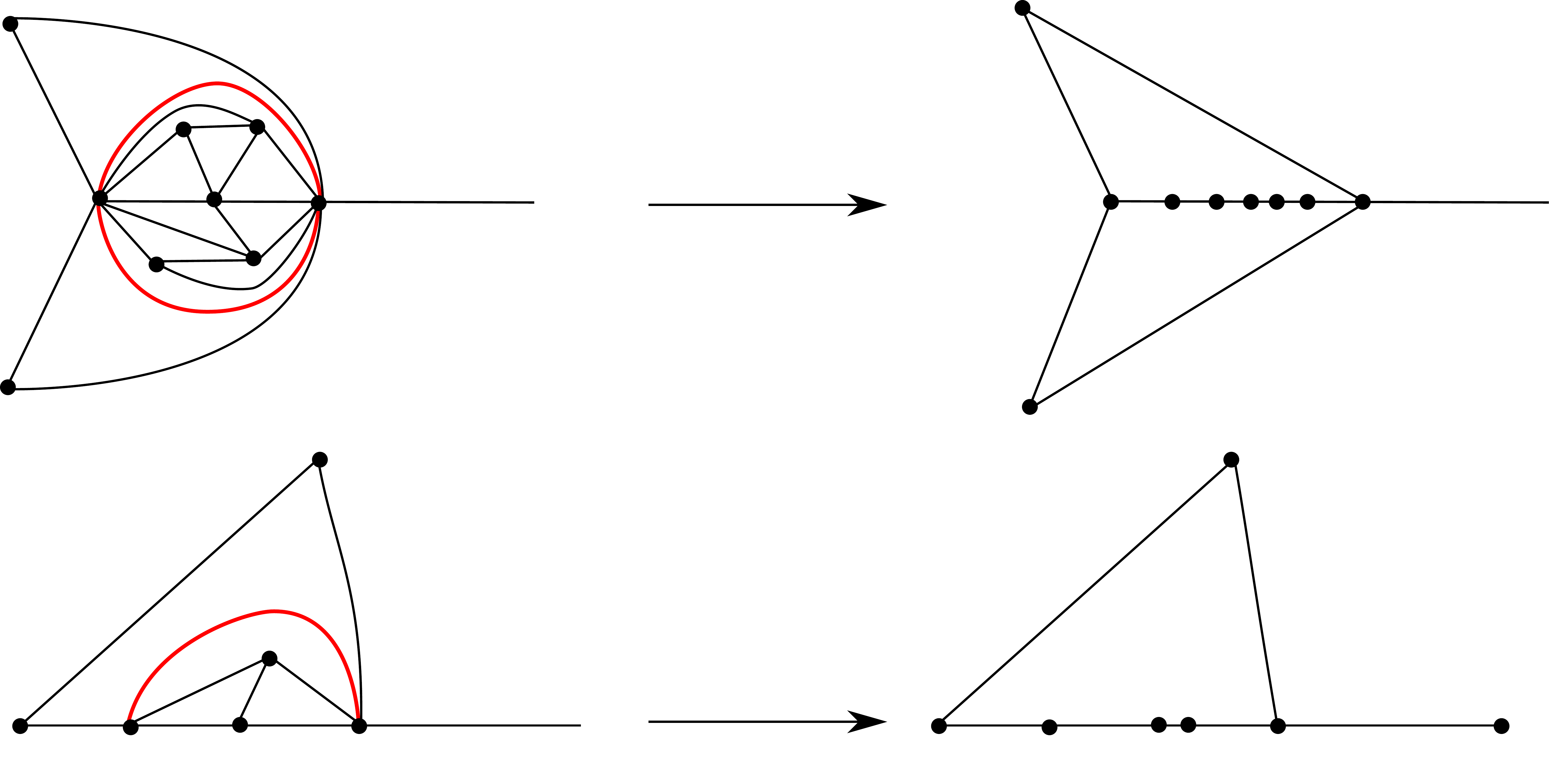}}%
    \put(0.28941283,0.00230949){\color[rgb]{0,0,0}\makebox(0,0)[lb]{\smash{$\partial X$}}}%
    \put(0.9167262,0.00230949){\color[rgb]{0,0,0}\makebox(0,0)[lb]{\smash{$\partial H$}}}%
    \put(0.46031806,0.34233554){\color[rgb]{0,0,0}\makebox(0,0)[lb]{\smash{$\varphi$}}}%
    \put(0.46031806,0.00230949){\color[rgb]{0,0,0}\makebox(0,0)[lb]{\smash{$\varphi$}}}%
  \end{picture}%
\endgroup%
	
  \caption{Any triangulation inducing a linear degeneracy would require either multiple edges (top) or a dividing edge (bottom).}
  \label{F:appendix}
\end{figure}
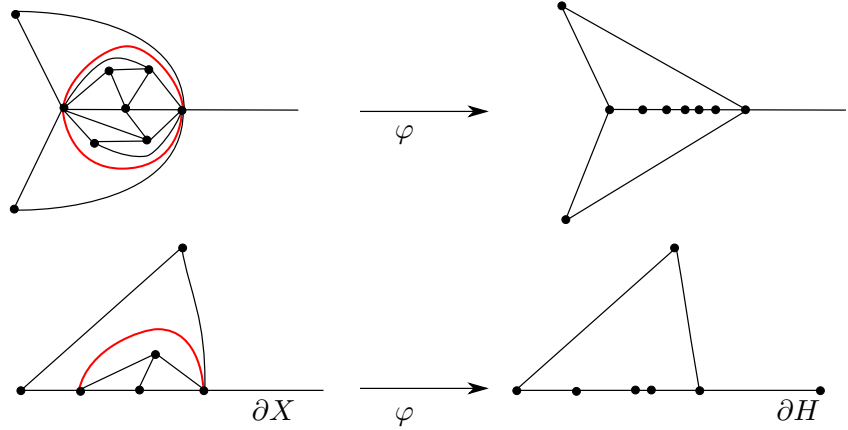


\textbf{Step 5: Conclusion.} Since $X_2$ is empty, no triangle in the
image of $\Phi$ is degenerate. Furthermore, all the $X_1$
are reduced to a single point and thus $K(v)$ is zero for all the
vertices $v$. The only remaining possible pathology is if all the
triangles adjacent to a non-boundary vertex $v$ are mapped to a
half-plane around $\Phi(v)$. By the convexity constraint, this can
only happen if the edges adjacent to $v$ are aligned, but this would
yield degenerate triangles. Therefore $\Phi$ is a local homeomorphism of degree $1$, hence it is a global homeomorphism and $\varphi$ is an embedding.
\fi 

\end{document}